\documentclass[usenames,dvipsnames,12pt]{article}

\usepackage{graphicx}
\usepackage{amsmath}
\usepackage{amssymb}
\usepackage{latexsym}
\usepackage{mathrsfs}
\usepackage{amsthm}
\usepackage{setspace}
\usepackage{epsfig}
\usepackage{subfigure}
\usepackage{authblk}
\usepackage{amsfonts}
\usepackage{url}
\usepackage{float}
\usepackage{multibib}
\usepackage{titling}
\newcites{appendix}{Reference}

\usepackage[authoryear,round]{natbib}	
\usepackage{mathtools}														 
\mathtoolsset{showonlyrefs=true}									 

\usepackage{geometry}
\usepackage[dvips]{color}
\newtheoremstyle{newplain}
{4pt}
{4pt}
{\itshape}
{}
{\itshape\bf}
{.}
{.5em}
{}
\theoremstyle{newplain}
\newtheorem{theorem}{Theorem}[section]
\newtheorem{proposition}{Proposition}[section]

\newtheorem{lemma}{Lemma}[section]
\newtheorem{definition}{Definition}[section]
\newtheorem{assumption}{Assumption}[section]

\newtheoremstyle{newdefinition}
{4pt}
{4pt}
{}
{}
{\itshape\bf}
{.}
{.5em}
{}
\theoremstyle{newdefinition}

\newtheorem{vg}{Example}[section]

\newtheorem{remark}{Remark}[section]
\numberwithin{equation}{section}
\newcommand{\be}{\begin{equation}}
\newcommand{\ee}{\end{equation}}
\newcommand{\nee}{\nonumber\end{equation}}
\newcommand{\eel}[1]{\label{#1}\end{equation}}
\newcommand{\brmk}[1]{\begin{remark}\label{#1}\begin{em} }
\newcommand{\ermk}{ $\quad\triangleleft$\end{em}\end{remark}}
\newcommand{\bvg}[1]{\begin{vg}\label{#1}\begin{em} }
\newcommand{\evg}{ $\quad\triangleleft$\end{em}\end{vg}}

\begin{document}
\bibliographystyle{plainnat}
\setlength{\abovedisplayskip}{8pt}
\setlength{\belowdisplayskip}{8pt}
\setlength{\abovedisplayshortskip}{4pt}
\setlength{\belowdisplayshortskip}{8pt}

\begin{titlepage}
\title{{\bf{\Huge Long Term Risk: A Martingale Approach}}
\thanks{The authors thank Lars Peter Hansen (the co-editor) and the anonymous referees for their insightful comments and suggestions that helped improve the paper, and Jaroslav Borovicka, Peter Carr, Timothy Christensen (discussant) and Jose Scheinkman for stimulating discussions.  This paper is based on research supported by the grant CMMI-1536503 from the National Science Foundation.}}
\author{Likuan Qin\thanks{likuanqin2012@u.northwestern.edu}
}
\author{Vadim Linetsky\thanks{linetsky@iems.northwestern.edu}
}
\affil{\emph{Department of Industrial Engineering and Management Sciences}\\
\emph{McCormick School of Engineering and Applied Sciences}\\
\emph{Northwestern University}}
\date{\today\\
\emph{To appear in Econometrica}}

\end{titlepage}

\maketitle

\begin{abstract}
This paper extends the long-term factorization of the stochastic discount factor introduced and studied by \citet{alvarez_2005using} in discrete-time ergodic environments  and by \citet{hansen_2009} and \citet{hansen_2012} in Markovian environments
to general semimartingale environments.
The transitory component discounts at the stochastic rate of return on the long bond and is factorized into discounting at the long-term yield and a positive semimartingale that extends the principal eigenfunction of \citet{hansen_2009} to the semimartingale setting. The permanent component is a  martingale that accomplishes a change of probabilities to the long forward measure, the limit of $T$-forward measures.
The change of probabilities from the data generating to the long forward measure absorbs the long-term risk-return trade-off and interprets the latter as the long-term risk-neutral measure.
\end{abstract}

\begin{center}
\end{center}

\section{Introduction}

Following \citet{alvarez_2005using}, \citet{hansen_2009} and \citet{hansen_2012},
this paper decomposes the arbitrage-free pricing kernel
$$
S_t=e^{-\lambda t}\frac{1}{\pi_t}M_t^\infty
$$
into discounting at the long-term discount rate
$\lambda$ (yield on the long bond, a pure discount bond maturing in the distant future), a process $\pi_t$ characterizing gross holding period returns on the long bond net of the long-term discount rate, and a positive martingale $M_t^\infty$ that defines a long-term forward measure. This measure absorbs the long-term risk adjustments of stochastically growing cash flows much like the risk-neutral measure absorbs short-term or instantaneous risk adjustments.
In contrast to the original operator approach of \citet{hansen_2009} and \citet{hansen_2012}, our martingale approach to the characterization of long-term pricing does not require a Markov specification and is based on a limiting procedure, constructing the long forward measure as the limit of finite maturity forward measures (\citet{jarrow_1987pricing}, \citet{jamshidian_1989exact}, \citet{geman_1995changes}) as maturity increases.
The long-term discount rate $\lambda$ and the process  $\pi_t$ are counterparts of the Perron-Frobenius eigenvalue and eigenfunction of  \citet{hansen_2009} and \citet{hansen_2012} in the sense that in Markovian economies the process $\pi_t$ reduces to the function of the Markovian state, $\pi(X_t)$, where $\pi(x)$ is the Perron-Frobenius eigenfunction of the pricing operator with the eigenvalue $e^{-\lambda t}$ as in \citet{hansen_2009}.


The paper is organized as follows. The key results of our long-term limit characterization are presented in Theorem \ref{implication_L1} and Theorem \ref{asym_long}.
Theorem \ref{long_holding} identifies the exponent $\lambda$ with the long-term yield on cash flows with bounded moments. Theorem \ref{long_power_yield} treats a degenerate case with $\lambda=0$ and shows that in this case, while the long-term discounting is sub-exponential, one may nevertheless define an asymptotic power yield.
Theorem \ref{long_growth_yield} features the long-term risk-return trade-off under the long forward measure.
It shows that under suitable moment conditions the limiting long-term yield  on cash flows with stochastic growth remains equal to the long-term yield on the pure-discount bond under the long forward measure, as the martingale component absorbs the long-term risk-return trade-off.
Thus, the long-term risk premia on cash flows vanish under the long forward measure even when the cash flows display stochastic growth.
This result leads to the interpretation of the long forward measure as the long-term risk-neutral measure and extends the corresponding Markovian result of \citet{borovicka_2014mis}.
Section \ref{example_Markov}
connects our results with the operator setting of \citet{hansen_2009}  and \citet{hansen_2012} in Markovian environments.

Our treatment based on semimartingale convergence naturally unifies discrete-time characterizations of  \citet{alvarez_2005using} and Markovian characterizations of \citet{hansen_2009},  \citet{hansen_2012} and \citet{borovicka_2014mis} in the framework of the martingale theory and
reveals that the long-term factorization is a fundamental feature of arbitrage-free asset pricing models, rather than an artifact of special assumptions, such as the Markov property.
This characterization enhances our understanding of risk pricing over alternative investment horizons.
The growing related literature includes
\citet{hansen_2012}, \citet{hansen_2012pricing},  \citet{hansen_2013}, \citet{borovicka_2014mis},
\citet{bakshi_2012}, \citet{christensen2014nonparametric},  \citet{christensen_2013estimating}, \citet{linetsky_2014_cont} and \citet{linetsky2016bond}.

\section{Semimartingale Pricing Kernels}
\label{setup}

We work on a complete filtered probability space $(\Omega,{\mathscr F},({\mathscr F}_{t})_{t\geq 0},{\mathbb P})$
with the filtration $({\mathscr F}_{t})_{t\geq 0}$ satisfying the usual conditions of right continuity and completeness. We assume that ${\mathscr F}_0$ is trivial modulo ${\mathbb P}$. The filtration models the information flow in continuous time. The conditional expectation $\mathbb{E}[\cdot|\mathscr{F}_t]$ is written as $\mathbb{E}_t[\cdot]$.
All random variables are identified up to almost sure equivalence. Stochastic processes which have the same paths outside a ${\mathbb P}$-null set are identified without further notice.
A stochastic process $(X_t)_{t\geq 0}$ is said to be {\em adapted} to the filtration $({\mathscr F}_{t})_{t\geq 0}$ if $X_t$ is measurable with respect to $\mathscr{F}_t$ for all $t\geq0$.
For a real-valued process $X$ with right-continuous with left limits (RCLL) paths, $X_{-}$ denotes the process of its left limits, $(X_-)_t=\lim_{s\uparrow t}X_{t}$ for $t>0$ and $(X_-)_0:=X_0$.
A {\em semimartingale} is a real-valued adapted RCLL process $X$ decomposable into the form $X_t=X_0+M_t+A_t,$ where $M_t$ is a local martingale (i.e. there exists a sequence of stopping times $(T_n)_{n\geq 1}$ increasing to infinity such that each stopped process $M_{t\wedge T_n}$ is a martingale)  and $A_t$ is a process of finite variation (i.e. whose paths have bounded variation over each finite time interval).
The semimartingale framework encompasses essentially all models in use in continuous-time finance, including models with stochastic volatility  and jumps. Moreover, discrete-time models are naturally embedded into continuous-time pure jump semimartingales with jumps at discrete times. We refer to \citet{jacod_1987limit} and \citet{protter_2003} for more details on semimartingales.

{\em Emery's distance}  between two semimartingales is defined by:
\be
d_{\cal S}^{\mathbb P}(X,Y)=\sum_{n\geq 1} 2^{-n} \sup_{|\eta|\leq 1}{\mathbb E}^\mathbb{P}\left[1 \wedge \left|\eta_0(X_0-Y_0)+\int_0^n \eta_s d(X-Y)_s\right| \right],
\eel{emeryd}
where $\int_0^t \eta_s dX_s$ denotes the stochastic integral of the predictable process $\eta$ with respect to the semimartingale $X$ and
the supremum is taken over all predictable processes $\eta$ bounded by one, $|\eta_t|\leq 1$.
A process $\eta_t$ is said to be {\em predictable} if it is measurable with respect to the $\sigma$-field on $\Omega\times\mathbb{R}_+$ generated by all left-continuous processes.
We can think of $X_t$ as the price of an asset and $\eta_t$ as the trading strategy (the number of units of the asset $X$ held at time $t$). Then the stochastic integral $\int_0^t \eta_s dX_s$ represents gains from the trading strategy up to time $t$. The predictable property of the trading strategy has the intuition that the agent cannot react instantaneously to the contemporaneous price change. That is, if $X_t$ has a surprise jump at time $t$, the agent cannot adjust his position at exactly the same time to profit from the jump. Endowed with Emery's metric the space of semimartingales is a complete topological vector space (\citet{emery_1979topologie}), and the corresponding topology is called Emery's semimartingale topology. It has a natural economic interpretation.
Suppose we have two asset price processes $X_t$ and $Y_t$ normalized so that $X_0=Y_0=1$.
For simplicity, consider a finite time horizon $[0,1]$ and all strategies trading $X_t$ or $Y_t$ with (long or short) positions restricted not to exceed one unit. Then $d^\mathbb{P}_{\cal S}(X,Y)$ measures the distance between the two assets in terms of maximum achievable difference from trading these two assets. If $X^n\xrightarrow{\rm {\cal S}}X$, where $\xrightarrow{\rm {\cal S}}$ denotes convergence in the semimartingale topology, then as $n$ increases, $X^n$ will become indistinguishable from $X$ in terms of gains from such trading strategies.
While Emery's distance depends on the probability measure under which the expectation is computed, Emery's semimartingale convergence is invariant under locally equivalent measure changes (Theorem II.5 in \citet{memin_1980espaces}).

We assume absence of arbitrage and trading frictions and existence of a {\em pricing kernel} (PK) process $S=(S_t)_{t\geq 0}$ satisfying the following assumptions.
\begin{assumption}
\label{ass_pricing}{\bf (Semimartingale Pricing Kernel)}
The pricing kernel process $S$ is a strictly positive semimartingale with $S_0=1$, $S_{-}$ is strictly positive, and ${\mathbb E}^{\mathbb P}\left[S_T/S_t\right]<\infty$ for all $T> t\geq0$.
\end{assumption}

Assumption \ref{ass_pricing} is in force throughout the paper without further mentioning. For each $0\leq t\leq T<\infty$ the PK defines a pricing operator $({\mathscr P}_{t,T})_{0\leq t \leq T}$ mapping  time-$T$ payoffs $Y$ (${\mathscr F}_T$-measurable random variables) into their time-$t$ prices
${\mathscr P}_{t,T}(Y)$ (${\mathscr F}_t$-measurable random variables):
${\mathscr P}_{t,T}(Y)={\mathbb E}_t^{\mathbb P}\left[S_T Y /S_t\right],$
where $S_T/S_t$ is the {\em stochastic discount factor} (SDF) from $T$ to $t$.

Having specified the PK $S$, we will be interested in the convex cone of positive semimartingales defined as follows.
\begin{definition}\label{def_value}{\bf (Valuation Processes Priced by $S$)}
A positive semimartingale
$V$ with the positive process of its left limits $V_{-}$ is said to be a {\em valuation process} if  the product $V_tS_t$ is a martingale.
\end{definition}
Valuation processes serve as models of assets priced by $S$ and are semimartingale counterparts of valuation functionals of  \citet{hansen_2009} in their Markovian setting. They include both capital gains and reinvested dividends, so that the gross total return earned from holding an asset with the valuation process $V$ during the period from $t$ to $T$ is given by $R_{t,T}^V=V_T/V_{t}$.

A $T$-maturity pure discount (zero-coupon) bond has a single unit cash flow at time $T$ and a valuation process
$P_{t}^T={\mathscr P}_{t,T}(1)={\mathbb E}_t^{\mathbb P}\left[S_T/S_t \right]$, $0\leq t\leq T.$
For each $T$ the zero-coupon bond valuation process $(P_{t}^T)_{t\in [0,T]}$ is a positive semimartingale such that $P_T^T=1$, and the process
\be
M_t^T:=S_t P_{t}^T/P_0^T
\eel{def_mtT}
is a positive martingale on $t\in [0,T]$ with $M_0^T=1$.
For each $T$ we can thus write the factorization $S_t =  (P_{0}^T/P_{t}^T)M_t^T$ on the time interval $t\in [0,T]$.

We can then use the martingale $M_t^T$ to define a new probability measure ${\mathbb Q}^T$ on ${\mathscr F}_{T}$ by
${\mathbb Q}^T|_{{\mathscr F}_{T}}=M_T^T {\mathbb P}|_{{\mathscr F}_{T}}.$
This is the $T$-{\em forward measure} (\citet{jarrow_1987pricing}, \citet{jamshidian_1989exact}, \citet{geman_1995changes}). Under ${\mathbb Q}^T$ the $T$-maturity zero-coupon bond serves as the numeraire, and the pricing operator reads:
${\mathscr P}_{s,t}(Y)=P_{s}^T\mathbb{E}_s^{\mathbb{Q}^T}\left[ Y/P_{t}^T \right]$
for an ${\mathscr F}_t$-measurable payoff $Y$ and $s\leq t\leq T$.

For a given $T$, the $T$-forward measure is defined on   ${\mathscr F}_{T}$ (and, hence, on ${\mathscr F}_{t}$ for all $t\leq T$). We now extend it to ${\mathscr F}_{t}$ for all $t> T$ as follows.
Fix $T$ and consider a self-financing roll-over strategy that starts at time zero by investing one unit of account in $1/P_{0}^T$ units of the $T$-maturity zero-coupon bond. At time $T$ the bond matures, and the value of the strategy is $1/P_{0}^T$ units of account. We roll the proceeds over by re-investing into $1/(P_{0}^T P_{T}^{2T})$ units of the zero-coupon bond with maturity $2T$. We continue with the roll-over strategy, at each time $kT$ re-investing the proceeds into the bond with maturity $(k+1)T$. We denote the valuation process of this self-financing strategy $B_t^T$:
\be
B_t^T = \left(\prod_{i=0}^k P_{iT}^{(i+1)T}\right)^{-1} P_{t}^{(k+1)T},\quad t\in [kT,(k+1)T),\quad k=0,1,\ldots.
\eel{Tbond}
It is clear by construction that the process $S_t B_t^T$ extends the martingale $M_t^T$ to all $t\geq 0$, and, thus, defines the $T$-forward measure ${\mathbb Q}^T$  on  ${\mathscr F}_{t}$ for all $t\geq 0$, where $T$ now has the meaning of the length of the compounding interval. Under  ${\mathbb Q}^T$ extended to all ${\mathscr F}_{t}$ with $t\geq 0$ in this manner, the roll-over strategy $(B_t^T)_{t\geq 0}$ with the compounding interval $T$ serves as the new numeraire. We continue to call the measure extended to all ${\mathscr F}_{t}$ for $t\geq 0$ the $T$-forward measure and use the same notation, as it reduces to the standard definition of the forward measure on ${\mathscr F}_{T}$.
Since the roll-over strategy $(B^T_t)_{t\geq 0}$ and the positive martingale $M_t^T=S_t B_t^T$ are now defined for all $t\geq 0$, we can write the $T$-forward  factorization of the pricing kernel for all $t\geq 0$ as
$S_t = (1/B_t^T)M_t^T.$ The positive martingale $M_t^T$ is now extended to all $t\geq0$ and is an extension of the previous definition (Eq. \eqref{def_mtT}).

\section{The Long-term Limit}
\label{long_term_section}
We are now ready to formally introduce and investigate the long-term factorization.

\begin{definition}{\bf (Long Bond)}
\label{def_longbond}
If the value processes $(B^T_t)_{t\geq 0}$ of the roll-over strategies in $T$-maturity bonds converge to a strictly positive semimartingale $(B_t^\infty)_{t\geq 0}$ uniformly on compacts in probability
as $T\rightarrow \infty$, i.e. for all $t>0$ and $K>0$
\be
\lim_{T\rightarrow \infty} {\mathbb P}(\sup_{s\leq t}|B_s^T-B_s^\infty|>K)=0,
\eel{longbonddef}
we call the limit the {\em long bond}.
\end{definition}

\begin{definition}{\bf (Long Forward Measure)}
\label{def_longforward}
If there exists a measure $\mathbb{Q}^\infty$ equivalent to $\mathbb{P}$ on each ${\mathscr F}_t$ such that the $T$-forward measures converge strongly to ${\mathbb Q}^\infty$ on each ${\mathscr F}_t$, i.e.
$\lim_{T\rightarrow \infty}{\mathbb Q}^T(A)={\mathbb Q}^\infty(A)$
for each $A\in {\mathscr F}_t$ and each $t\geq 0$,
we call it the {\em long forward measure} and denote it ${\mathbb L}$.
\end{definition}

The following theorem gives an explicit sufficient condition easy to verify in applications that ensures stronger modes of convergence --- Emery's semimartingale convergence of valuation processes $B^T$ to the long bond and convergence in total variation of the $T$-forward measures ${\mathbb Q}^T$ to the long forward measure.

\begin{theorem}{\bf (Long Term Factorization and the Long Forward Measure)}
\label{implication_L1}
Suppose that for each $t>0$ the ratio of the ${\mathscr F}_t$-conditional expectation of the pricing kernel $S_T$ to its unconditional expectation converges to a positive limit in $L^1$ as $T\rightarrow \infty$ (under ${\mathbb P}$), i.e. for each $t>0$ there exists an almost surely positive ${\mathscr F}_t$-measurable random variable which we denote $M_t^\infty$ such that
\be
\frac{{\mathbb E}^{\mathbb P}_t[S_T]}{{\mathbb E}^{\mathbb P}[S_T]} \xrightarrow{\rm L^1} M_t^\infty\quad \text{as} \quad T\rightarrow \infty.
\eel{PKL1}
Then the following results hold:\\
(i) The collection of random variables $(M_t^\infty)_{t\geq0}$ is a positive ${\mathbb P}$-martingale, and the family of martingales $(M_t^T)_{t\geq 0}$ converges to  the martingale $(M_t^\infty)_{t\geq0}$ in the semimartingale topology.\\
(ii) The long bond valuation process $(B_t^\infty)_{t\geq0}$ exists, and the roll-over strategies $(B_t^T)_{t\geq 0}$ converge to the long bond $(B_t^\infty)_{t\geq 0}$ in the semimartingale topology.\\
(iii) The pricing kernel possesses the long-term factorization
\be
S_t=\frac{1}{B_t^\infty}M_t^\infty.
\eel{ltf}
(iv) $T$-forward measures ${\mathbb Q}^T$ converge to the long forward measure ${\mathbb L}$ in total variation  on each ${\mathscr F}_t$, and ${\mathbb L}$ is equivalent to ${\mathbb P}$ on ${\mathscr F}_t$ with the Radon-Nikodym derivative $M_t^\infty$.
\end{theorem}

The proof is given in Appendix \ref{proof_l1}. Theorem \ref{implication_L1} makes the economics of the long forward measure clear. Since under ${\mathbb L}$ the pricing kernel reduces to the reciprocal of the long bond, the long bond is growth optimal under $\mathbb{L}$ (cf. \citet{bansal1997growth}). \citet{linetsky2016bond} further show that under $\mathbb{L}$ the term structure of bond Sharpe ratios
for a sufficiently small holding period generally has an increasing shape in the bond maturity $T$, with
the long bond achieving the maximal instantaneous Sharpe ratio under ${\mathbb L}$ (the \citet{hansen_1990implications} bound). The empirical shape of the term structure of bond Sharpe ratios estimated in \citet{linetsky2016bond} is generally opposite to the one described above, indicating that the martingale component in the long-term factorization is highly economically significant, complementing empirical results in \citet{alvarez_2005using}, \citet{bakshi_2012}, \citet{borovicka_2014mis} and \citet{christensen2014nonparametric}.

We note that the condition Eq.\eqref{PKL1} does not restrict the asymptotic behavior of the initial term structure $P_0^T={\mathbb E}^{\mathbb P}[S_T]$ as maturity $T$ increases, but restricts the time evolution of the PK so that the asymptotic behavior of the initial term structure $P_0^T$ is preserved as time goes on in the sense that for each $t>0$ the ratio $P_t^T/P_0^T$ possesses a finite positive limit.
We next show that under an assumption imposing regularity on the asymptotic behavior of the initial term structure $P_0^T$ added to the assumption Eq.\eqref{PKL1} we can achieve a more refined characterization of the long term factorization, where we can further factorize the long bond $B^\infty_t$ into an exponential factor $e^{\lambda t}$ and a semimartingale $\pi_t$ that extends the principal eigenfunction of \citet{hansen_2009} in Markovian environments to general semimartingale environments.
To this end, we first recall the definition of slowly varying functions. A measurable function $L: (0,\infty)\rightarrow (0,\infty)$ is called {\em slowly varying} (at infinity) if for all $a>0$ the ratio $L(ax)/L(x)$ converges to one as $x\rightarrow\infty$. If this limit is a finite positive number for each $a>0$, but not necessarily equal to one, the function is called {\em regularly varying} (see \citet{bingham1989regular} for a detailed study of slowly varying functions).

\begin{theorem}{\bf (Long Term Factorization of the Long bond)}
\label{asym_long}
Suppose assumption Eq.\eqref{PKL1} in Theorem \ref{implication_L1} holds and in addition suppose that for each $t>0$ the ratio $P_0^{T-t}/P_0^T$ has a positive finite limit as $T\rightarrow \infty$.
Then the following results holds:\\
(i) There exists a constant $\lambda$ such that for each $t>0$
$$\lim_{T\rightarrow \infty}\frac{P_0^{T-t}}{P_0^T} = e^{\lambda t}.$$
Furthermore, there exists a slowly varying function $L(x)$ such that $$P_0^t=e^{-\lambda t}L(e^t),\quad t\geq 0.$$
(ii) For all $t\geq0$,
\be
\lim_{T\rightarrow\infty}\frac{-\log P_t^T}{T-t}=\lambda,
\eel{lt_rate}
where the limit is in probability under any measure locally equivalent to $\mathbb{P}$.\\
(iii) The sequence of semimartingales $(\pi_t^T)_{t\geq 0}$ defined for each $T>0$ by $\pi_t^T:=P_t^T/P_0^{T-t}$ for $t\leq T$ and $\pi_t^T:=1$ for $t>T$ converges to a positive semimartingale $\pi_t$ with $\pi_0=1$ in the semimartingale topology as $T\rightarrow \infty$.\\
(iv)
The long bond possesses a factorization $B_t^\infty=e^{\lambda t}\pi_t,$
so that the long-term factorization of the pricing kernel reads
\be
S_t = e^{-\lambda t} \frac{1}{\pi_t} M_t^\infty.
\eel{ltpi}
(v) The positive semimartingale $\pi_t$ satisfies:
\be
\mathbb{E}_t^\mathbb{P}\left[\frac{S_T}{S_t}\pi_T\right]=e^{-\lambda(T-t)}\pi_t
\eel{eigens}
for all $0\leq t<T$ and
\be
\lim_{T\rightarrow\infty}\frac{1}{T-t}\log(\mathbb{E}_t^\mathbb{L}[1/\pi_T])=0
\eel{Lmean}
for all $t\geq 0$, where the limit is in probability under any measure locally equivalent to $\mathbb{P}$.
\end{theorem}
The proof is given in Appendix \ref{proof_l1} and relies on Karamata's Characterization Theorem for regularly varying functions that states that any regularly varying function is of the form $x^{-\lambda} L(x)$ for some real constant $\lambda$ and a slowly varying function $L(x)$.  We note that the original long-term characterization of \citet{alvarez_2005using} in discrete time is naturally nested in Theorems \ref{implication_L1} and \ref{asym_long} as a special case by embedding a discrete-time adapted process into a continuous-time semimartingale. Appendix \ref{app_aj} provides the precise result.

Theorem \ref{asym_long} shows that under the regularity assumption on the asymptotic behavior of the initial term structure requiring convergence of the ratio $P_0^{T-t}/P_0^T={\mathbb E}^{\mathbb P}[S_{T-t}]/{\mathbb E}^{\mathbb P}[S_{T}]$ as $T\rightarrow \infty$ for each fixed $t$ along with our assumption Eq.\eqref{PKL1} that the asymptotic behavior of the term structure is preserved over time,  the pricing kernel possesses a positive semimartingale $\pi_t$ that directly extends the process $\pi(X_t)$ associated with the principal eigenfunction $\pi(x)$ of \citet{hansen_2009} in Markovian environments (see Section \ref{example_Markov} for details).
Indeed, Eq.\eqref{eigens} directly extends the eigenfunction problem studied in \citet{hansen_2009} and \citet{linetsky_2014_cont} (also see Section \ref{example_Markov}). Eq.\eqref{Lmean} shows that, after removing the exponential growth or decay, $e^{\lambda t}$, the $\mathbb{L}$-mean of the reciprocal of the long bond has zero growth rate. Thus, the factor $e^{\lambda t}$, in fact, removes {\em all} of the exponential growth or decay, and our factorization $B_t^\infty=e^{\lambda t}\pi_t$ is indeed germane to the study of the long-term behavior of the pricing kernel.
The corresponding long-term factorization \eqref{ltpi}, refining the factorization \eqref{ltf}, is a semimartingale counterpart of the long-term factorization of  \citet{hansen_2009} associated with the principal eigenvalue of the pricing kernel germane to its long-term behavior (see Section \ref{example_Markov} for further details).

Eq.\eqref{lt_rate} in Theorem \ref{asym_long} also implies that $\lambda$ appearing in the long-term factorization is the long-term discount rate (the long-term asymptotic zero-coupon bond yield) and is independent of time $t$ when the yield is computed. This is consistent with the theorem of \citet{dybvig_1996long}, who assert that the long zero-coupon rate can never fall under more general circumstances. Under our conditions in Theorem 3.2, the long rate remains constant, rather than merely non-decreasing.

As we show next, the property Eq.\eqref{Lmean} of the semimartingale $\pi_t$ turns out to be essential for the study of the long-term risk-return trade-off under $\mathbb{L}$.
%
To this end, we consider  a positive semimartingale cash flow process $(C_t)_{t\geq 0}$. The $\mathbb{L}$-expected gross return over the holding period from time $t$ to time $T$ from receiving the cash flow $C_T$ at time $T$ is
$\mathbb{E}_t^\mathbb{L}[C_T]/\mathbb{E}_t^\mathbb{P}[S_TC_T]=\mathbb{E}_t^\mathbb{L}[C_T]/\mathbb{E}_t^\mathbb{L}[C_T/B_T^\infty].$
Following Section IV.A in  \citet{borovicka_2014mis}, we also define the expected exponential yield under ${\mathbb L}$ as follows:
\be
\rho_{t,T}^{\mathbb L}(C_T):=\frac{1}{T-t}\log\left(\frac{\mathbb{E}_t^\mathbb{L}[C_T]}{\mathbb{E}_t^\mathbb{L}[C_T/B_T^\infty]}\right)=\lambda+\frac{1}{T-t}\log\left(\frac{\mathbb{E}_t^\mathbb{L}[C_T]}{\mathbb{E}_t^\mathbb{L}[C_T/\pi_T]}\right).
\eel{rho_exp}
We have the following result characterizing the asymptotic yield.
\begin{theorem}{\bf (Long-term Exponential Yield)}
\label{long_holding}
(i) Suppose the assumptions of Theorem \ref{asym_long} hold, and for some $t\geq 0$ there exist positive constants $0<c<C<\infty$ and $T'>0$ such that almost surely
$c<\mathbb{E}_t^\mathbb{L}[C_T]<C$  and $c<\mathbb{E}_t^\mathbb{L}[C_T/\pi_T]/\mathbb{E}_t^\mathbb{L}[1/\pi_T]<C$ for all $T>T'$. Then we have
\be
\lim_{T\rightarrow\infty}\rho_{t,T}^{\mathbb L}(C_T)=\lambda
\ee
where the limit is in probability under any measure locally equivalent to $\mathbb{P}$.\\
(ii) If furthermore there exist $0<c<C<\infty$ such that almost surely $c<\mathbb{E}_t^\mathbb{P}[C_T]<C$ for all $T>T'$, then $\rho_{t,T}^{\mathbb P}(C_T)$, the expected exponential yield under ${\mathbb P}$, has the same asymptotic limit $\lambda$.
\end{theorem}

The proof is based on \eqref{Lmean} and is given in Appendix \ref{proof_l1}.
Theorem \ref{long_holding} shows that as long as the appropriate moments of the cash flow process $C_T$ remain bounded as $T$ increases,
the asymptotic exponential yield on the cash flow is equal to the long-term zero-coupon yield $\lambda$, regardless of the specifics of the cash flow process.
One class of examples of positive cash flow processes satisfying the bounded moments assumptions in Theorem \ref{long_holding} are bounded cash flows that are also bounded below away from zero. Another important class of examples are cash flows of the form $C_t=f(X_t)$, where $X_t$ is a Markov state satisfying appropriate stability assumptions under ${\mathbb L}$ and the payoff function $f$ satisfying appropriate moment conditions (see \citet{borovicka_2014mis} for such examples). If the ${\mathbb P}$-moments are also bounded and bounded away from zero, then $\lambda$ is also the limiting yield under ${\mathbb P}$, irrespective of the structure of the cash flow process $C_t$.
In other words, as in the Markovian setting in \citet{borovicka_2014mis}, appropriately bounded or stationary cash-flow risk does not alter the long-term yield. The limiting risk premium is zero under both probability measures ${\mathbb P}$ and ${\mathbb L}$
for such cash-flow risks.

We note that for pricing kernels with $\lambda=0$, the limiting result in Theorem \ref{long_holding} degenerates as the limiting exponential yield vanishes, since in this case discounting at the exponential rate is too fast. In particular, consider the case where $P_0^t=O(t^{-\gamma})$ (see \citet{brody_2013social} for their model of social discounting). In this case we have a similar limiting result for the power yield. Define the expected power yield as follows:
\be
\varrho_{t,T}^{\mathbb L}(C_T):=\frac{1}{\log (T-t)}\log\left(\frac{\mathbb{E}_t^\mathbb{L}[C_T]}{\mathbb{E}_t^\mathbb{L}[C_T/B_T^\infty]}\right)
\eel{rho_power}
and similarly under ${\mathbb P}$.
\begin{theorem}{\bf (Long-term Power Yield)}
\label{long_power_yield}
(i) Suppose the assumptions of Theorem \ref{asym_long} hold and $P_0^t=O(t^{-\gamma})$ as $t\rightarrow\infty$. Suppose further that for some $t\geq 0$ there exist $0<c<C<\infty$ and $T'>0$ such that almost surely $c<\mathbb{E}_t^\mathbb{L}[C_T]<C$ and $c<\mathbb{E}_t^\mathbb{L}[C_T/\pi_T]/\mathbb{E}_t^\mathbb{L}[1/\pi_T]<C$ for all $T>T'$. Then we have
\be
\lim_{t\rightarrow\infty}\varrho_{t,T}^{\mathbb L}(C_T)=\gamma
\ee
where the limit is in probability under any measure locally equivalent to $\mathbb{P}$.\\
(ii) If furthermore there exist $0<c<C<\infty$ such that almost surely
$c<\mathbb{E}_t^\mathbb{P}[C_T]<C$ for all $T>T'$, then $\varrho_{t,T}^{\mathbb P}(C_T)$, the expected power yield under ${\mathbb P}$, has the same asymptotic limit $\gamma$.
\end{theorem}
Theorems \ref{long_holding} and \ref{long_power_yield} consider cash flow processes with moments that remain bounded and, thus, exclude  long-term growth. Following Hansen and Scheinkman (2009) and \citet{borovicka_2014mis}, Section IV, we now consider
cash flows whose stochastic growth implies non-vanishing limiting risk premia under ${\mathbb P}$.
Namely, consider a positive semimartingale {\em growth index} $G_t$ (normalized by $G_0=1$) that can be interpreted as the inflation index when modeling inflation-indexed bonds or aggregate equity dividend growth when modeling equity.
We are interested in the exponential yield on the stochastically growing cash flow $G_t$ under ${\mathbb L}$, i.e.
\be
\rho_{t,T}^{\mathbb L}(G_T)=\frac{1}{T-t}\log\left(\frac{\mathbb{E}_t^\mathbb{L}[G_T]}{\mathbb{E}_t^\mathbb{L}[G_TB_t/B_T^\infty]}\right)=\lambda +\frac{1}{T-t}\log\left(\frac{\mathbb{E}_t^\mathbb{P}[\frac{S_TG_T\pi_T}{S_t\pi_t}]}{\mathbb{E}_t^\mathbb{P}[\frac{S_TG_T}{S_t}]}\right).
\eel{growth_yield}
In the second equality we re-wrote the ${\mathbb L}$-expectations in terms of ${\mathbb P}$-expectations using the relationship $M_t^\infty=S_t\pi_t e^{\lambda t}$.
If $G_t$ is interpreted as the inflation index, we can treat $S_tG_t$ as the real pricing kernel. More generally, $S_tG_t$ is interpreted as the growth-indexed pricing kernel.
If we assume that $S_tG_t$ also satisfies the conditions in Theorem 3.1 and Theorem \ref{asym_long}, i.e. for each $t>0$ $\mathbb{E}_t^\mathbb{P}[S_TG_T]/\mathbb{E}^\mathbb{P}[S_TG_T]$ converges to a positive random variable in $L^1$ as $T\rightarrow \infty$ and $\mathbb{E}^\mathbb{P}[S_{T-t}G_{T-t}]/\mathbb{E}^\mathbb{P}[S_TG_T]$ converges to a positive finite limit, then we have the long-term factorization for the growth-indexed pricing kernel $$S_tG_t=e^{-\lambda^G t}\frac{1}{\pi_t^G}M_t^{G,\infty},$$ where $M_t^{G,\infty}$ is a martingale and can be used to define a new probability measure ({\em long  forward growth measure}) $\mathbb{G}|_{\mathscr{F}_t}=M_t^{G,\infty}\mathbb{P}|_{\mathscr{F}_t}$ on each $\mathscr{F}_t$. We then have the following result paralleling the long-term risk-return trade-off under ${\mathbb L}$ formulated in \citet{borovicka_2014mis} in the Markovian setting.

\begin{theorem}{\bf (Long-Term Risk-Return Trade-off under $\mathbb{L}$)}
\label{long_growth_yield}
Suppose both the pricing kernel $S_t$ and the growth-indexed pricing kernel $S_tG_t$ satisfy the conditions in Theorem \ref{asym_long}. Suppose further that for some $t\geq 0$ there exist constants $0<c<C<\infty$ and $T'>0$ such that almost surely $c<\frac{\mathbb{E}_t^{\mathbb{G}}[\pi_T/\pi_T^G]}{\mathbb{E}_t^{\mathbb{G}}[1/\pi_T^G]}<C$ for all $T>T'$. Then
\be
\lim_{t\rightarrow\infty}\rho_{t,T}^{\mathbb L}(G_T)=\lambda,
\eel{growth_limit_yield}
where the limit is in probability under any measure locally equivalent to $\mathbb{P}$.
\end{theorem}
This result shows that under suitable moment conditions the limiting yield under the long forward measure ${\mathbb L}$ remains the same and equal to the long-term yield on the pure-discount bond even after we introduce stochastic growth in the cash flow. Thus, the long-term risk premia on cash flows vanish under ${\mathbb L}$ even when the cash flows display stochastic growth. Evidently, this conclusion is altered under ${\mathbb P}$, where the limiting yield of the growing cash flow $G_t$ is no longer equal to $\lambda$. This result leads to the interpretation of the long forward measure as the {\em long-term risk-neutral measure}.
As such, Theorem \ref{long_growth_yield} extends the result in Section IV.B of \citet{borovicka_2014mis} from their Markovian setting to our general semimartingale setting.

\section{Markovian Environments}
\label{example_Markov}


Our focus here is to show how results of \citet{hansen_2009} in Markovian environments based on their Perron-Frobenius theory of positive eigenfunctions of Markovian pricing operators naturally arise in the context of Theorem \ref{implication_L1} and \ref{asym_long}.
We now assume that the underlying filtration is generated by a Markov process $X$ and the PK is a positive  multiplicative functional of $X$.
More precisely,
the stochastic driver of all economic uncertainty is assumed to be a  {\em conservative Borel right process} (BRP) $X=(\Omega,{\mathscr F},({\mathscr F}_{t})_{t\geq 0},(X_t)_{t\geq 0},({\mathbb P}_x)_{x\in E})$.
A BRP is a continuous-time, time-homogeneous Markov process taking values in a Borel subset $E$ of some metric space (so that $E$ is equipped with the Borel sigma-algebra ${\mathscr E}$; in applications we can think of $E$ as a Borel subset of the Euclidean space ${\mathbb R}^d$), having right-continuous paths and possessing the strong Markov property (i.e., the Markov property extended to stopping times). The probability measure  ${\mathbb P}_x$ governs the behavior of the process $(X_t)_{t\geq 0}$ when started from $X_0=x\in E$ at time zero.
If the process starts from a probability distribution $\mu$, the corresponding measure is denoted by ${\mathbb P}_\mu$.
A statement concerning $\omega\in \Omega$ is said to hold ${\mathbb P}$-almost surely if it is true ${\mathbb P}_x$-almost surely for all $x\in E$. The filtration  $({\mathscr F}_{t})_{t\geq 0}$  in our model is the filtration generated by $X$ completed with ${\mathbb P}_\mu$-null sets for all initial distributions $\mu$ of $X_0$. It is right continuous. $X$ is assumed to be conservative, i.e. ${\mathbb P}_x(X_t\in E)=1$ for each initial $x\in E$ and all $t\geq 0$ (no killing or explosion).

 \citet{cinlar_1980} show that stochastic calculus of semimartingales defined over a BRP can be set up so that all key properties hold simultaneously for all starting points $x\in E$ and, in fact, for all initial distributions $\mu$ of $X_0$.
In particular, an $({\mathscr F}_{t})_{t\geq 0}$-adapted process $S$ is an ${\mathbb P}_x$-semimartingale (local martingale, martingale) {\em simultaneously for all} $x\in E$ and, in fact, for all ${\mathbb P}_\mu$.
We often simply write ${\mathbb P}$ where, in fact, we are dealing with the family of measures $({\mathbb P}_x)_{x\in E}$ indexed by the initial state $x$. Correspondingly, we simply say that a process is a ${\mathbb P}$-semimartingale (local martingale, martingale), meaning that it is a ${\mathbb P}_x$-semimartingale (local martingale, martingale) for each $x\in E$.


In this section we make some additional assumptions about the pricing kernel.
\begin{assumption}{\bf (Markovian Pricing Kernel)}
\label{markov_pricing}
The PK $(S_t)_{t\geq 0}$ is assumed to be a positive semimartingale multiplicative functional of $X$, i.e.
$S_{t+\tau}(\omega)=S_t(\omega)S_\tau(\theta_t(\omega)),$
where $\theta_t: \Omega\rightarrow\Omega$ is the shift operator (i.e. $X_\tau(\theta_t(\omega))=X_{t+\tau}(\omega)$), the process of its left limits is assumed to be positive, $S_{-}>0$, $S$ is normalized so that $S_0=1$, and ${\mathbb E}^{\mathbb P}_x[S_t]<\infty$ for all times $t>0$ and every initial state $x\in E$.
\end{assumption}

Under Assumption \ref{markov_pricing} the time-$t$ price of a payoff $f(X_{t+\tau})$ at time $t+\tau$ that depends on the Markovian state at that time can be written in the form:
\be
{\mathbb E}^{\mathbb P}\left[S_{t+\tau} f(X_{t+\tau})/S_t| {\mathscr F}_t\right]={\mathbb E}^{\mathbb P}_{X_t}[S_{\tau}f(X_{\tau})]={\mathscr P}_{\tau}f(X_t),
\eel{Markovianpricing}
where we used the Markov property and time homogeneity of $X$ and the multiplicative property of $S$ and introduced a family of {\em pricing operators} $({\mathscr P}_t)_{t\geq 0}$:
${\mathscr P}_{t}f(x):={\mathbb E}_{x}^{\mathbb P}[S_t f(X_t)],$
where ${\mathbb E}_{x}^{\mathbb P}$ denotes the expectation with respect to  ${\mathbb P}_x$.
The pricing operator ${\mathscr P}_t$ maps the payoff $f$ as a function of the state at the payoff time $t$ into its present value at time zero as a function of the initial state $X_0=x$.

Suppose the pricing operators $\mathscr{P}_t$ possess a positive eigenfunction $\pi$ satisfying
\be
\mathscr{P}_t\pi(x)=e^{-\lambda t}\pi(x)
\ee
for  some $\lambda\in\mathbb{R}$ and all $t>0$ and $x\in E$.
The key insight of \citet{hansen_2009} is that then the
pricing kernel admits a factorization
\be
S_t=M_t^\pi e^{-\lambda t}\pi(X_0)/\pi(X_t)
\eel{eigenfactorization}
into a transition-independent multiplicative functional $e^{-\lambda t}\pi(X_0)/\pi(X_t)$ and a positive martingale multiplicative functional $M_t^\pi=S_t e^{\lambda t}\pi(X_t)/\pi(X_0)$ of $X$.
Since $M_t^\pi$ is a positive $\mathbb{P}$-martingale starting from one, \citet{hansen_2009} define a new probability measure $\mathbb{Q}^\pi$  associated with the eigenfunction $\pi$ by:
$\mathbb{Q}^\pi|_{\mathscr{F}_t}=M_t^\pi\mathbb{P}|_{\mathscr{F}_t}.$
The pricing operator can then be expressed as:
\be
\mathscr{P}_t f(x)=e^{-\lambda t}\pi(x)\mathbb{E}_x^{\mathbb{Q}^\pi}\left[f(X_t)/\pi(X_t)\right],
\eel{pricingopqpi}
where $\mathbb{E}_x^{\mathbb{Q}^\pi}$ is the expectation with respect to the {\em eigen-measure} $\mathbb{Q}^\pi_x$ with $X_0=x$. 


Our key result is that, if the long forward measure ${\mathbb L}$ exists in a Markovian environment, then, under some regularity assumptions, it is necessarily identified with one of the eigen-measures. This result naturally links our Theorem \ref{implication_L1} and Theorem \ref{asym_long} in semimartingale environments with the Perron-Frobenius theory of Markovian pricing operators of \citet{hansen_2009}. Specifically, we show that in Markovian environments in the long-term factorization \eqref{ltpi} the semimartingale $\pi_t$ takes the form $\pi_t=\frac{\pi(X_t)}{\pi(X_0)},$ where $\pi(x)$ is a positive eigenfunction of the pricing operator ${\mathscr P}_t$ with the eigenvalue $e^{-\lambda t}$.

We start with the observation that, by Eq.\eqref{Markovianpricing}, in a Markovian setting the zero-coupon bond valuation process can be written in terms of the Markov state as follows,
$P_t^T=({\mathscr P}_{T-t}1)(X_t)=P(T-t,X_t)$ for $0\leq t\leq T$,
where $P(t,x)={\mathbb E}_x^{\mathbb P}[S_t]$ is the bond pricing function.
We also note that the valuation functionals $B_t^T(x)$ tracking gross returns from time zero to time $t$ on investing one unit of account at time zero in pure discount bonds and rolling over as in Section 2 now depend on the state $X_0=x$.

In Appendix \ref{appendix_markov}, it is shown that when Eq.\eqref{PKL1} holds under $\mathbb{P}_x$ for each $x\in E$, then $P(T-t, X_t)/P(T,x)$ converges in probability under $\mathbb{P}_x$ for each $x$. The following theorem shows that if we strength convergence in probability to pointwise convergence, we have the desired result.

\begin{theorem} {\bf (Long Forward Measure as an Eigen-Measure)}
\label{L_equal_eigen}
Suppose the pricing kernel satisfies Assumption \ref{markov_pricing} and Eq.\eqref{PKL1} holds under ${\mathbb P}_x$ for
each initial state $x\in E$. In addition suppose that the function $P(T-t,y)/P(T,x)$ converges to a positive limit as $T\rightarrow \infty$ for each fixed $t>0$ and $x,y\in E$.
Then the long bond valuation functional is identified with a positive multiplicative functional of $X$ in the transition independent form:
\be
B_t^\infty(x)=e^{\lambda_L t}\pi_L(X_t)/\pi_L(x),
\eel{longbondmarkovian}
where $\pi_L(x)$ is a positive eigenfunction of the pricing operators
$({\mathscr P}_t)_{t\geq 0}$ with the eigenvalues $e^{-\lambda_L t}$ for some $\lambda_L \in {\mathbb R}$,
the pricing kernel possesses a Hansen-Scheinkman eigen-factorization Eq.\eqref{eigenfactorization}, and the long forward measure ${\mathbb L}$ is identified with the Hansen-Scheinkman eigen-measure ${\mathbb Q}^{\pi_L}$.
\end{theorem}
The proof is given in Appendix C and relies on the Markov property, measurability properties of the bond pricing function, and Cauchy's multiplicative functional equation.
This theorem refines Theorem \ref{asym_long} in Markovian environments by further identifying the process $\pi_t$ with $\pi_L(X_t)/\pi_L(X_0),$ where $\pi_L(x)$ is a positive eigenfunction of the pricing operator.

Positive eigenfunctions of Markovian pricing operators are in general not unique.
For a given pricing kernel $S$ satisfying Assumption \ref{markov_pricing} there may be other positive eigenfunctions not associated with the long-term factorization. If the researcher is faced with the problem of determining the long-term factorization by first finding the eigenfunction $\pi_L$, it is useful to have sufficient conditions that single out $\pi_L$ from among all positive eigenfunctions.  \citet{hansen_2009}, \citet{borovicka_2014mis} and \citet{linetsky_2014_cont} show that if the state process $X$ satisfies certain stochastic stability assumptions under the eigen-measure  ${\mathbb Q}^\pi$, then the corresponding eigenfunction is unique.
Moreover, under sufficiently strong stability assumptions it is possible to identify this eigenfunction with the one germane to the long-term factorization.

The weakest among stochastic stability assumptions is recurrence.
 \citet{linetsky_2014_cont} establish that there exists at most one positive eigenfunction such that $X$ is recurrent under $\mathbb{Q}^{\pi}$.
If a positive eigenfunction $\pi$ such that $X$ is recurrent under ${\mathbb Q}^\pi$ exists, \citet{linetsky_2014_cont} call it {\em recurrent eigenfunction}, denote it and the corresponding eigenvalue by $\pi_R$ and $\lambda_R$, and call the associated eigen-measure ${\mathbb Q}^{\pi_R}$ {\em recurrent eigen-measure}.
While the recurrent eigen-measure is unique, recurrence is generally too weak to ensure identification of $\pi_R$ with $\pi_L$ and, hence, the recurrent eigen-measure with the long forward measure.
The following assumption is sufficient for the identification of the recurrent eigen-measure with the long forward measure.
\begin{assumption}
\label{exp_ergo_assumption}{\bf (Exponential Ergodicity)}
Suppose the Markovian pricing kernel satisfying Assumption \ref{markov_pricing} admits a recurrent eigenfunction $\pi_R$.
Suppose further there exists a probability measure $\varsigma$ on $E$ and some positive constants
$c,\alpha, t_0$ such that under $\mathbb{Q}^{\pi_R}$ the following {\em exponential ergodicity estimate} holds for all Borel functions satisfying $|f|\leq 1$:
\be
\left|\mathbb{E}^{\mathbb{Q}^{\pi_R}}_x\left[f(X_t)/\pi_R(X_t)\right]-c_f\right|
\leq c e^{-\alpha t}/\pi_R(x)
\eel{exp_ergo}
for all $t\geq t_0$ and each $x\in E$, where
$c_f={\mathbb E}^\varsigma[f(Y)/\pi_R(Y)]=\int_E (f(y)/\pi_R(y))\varsigma(dy)$.
\end{assumption}
Sufficient conditions for Eq.\eqref{exp_ergo} for Borel right processes can be found in Theorem 6.1 of \citet{meyn_1993stability}.
\begin{theorem}{\bf (Identification of ${\mathbb L}$ and ${\mathbb Q}^{\pi_R}$ and Long-Term Pricing)}
\label{markov_long}
If Assumption \ref{exp_ergo_assumption} is satisfied, then Eq.\eqref{PKL1} holds with $M_t^\infty=M_t^{\pi_R}$, Theorem \ref{implication_L1} applies, the long bond is given by Eq.\eqref{longbondmarkovian} with $\pi_L=\pi_R$,
and the recurrent eigen-measure coincides with the long forward measure, $\mathbb{Q}^{\pi_R}=\mathbb{L}.$ Moreover, the long-term pricing formula of \citet{hansen_2009} holds for any bounded payoff $f(X_t)$ at time $t\geq t_0$ with the following error estimate:
\be
\left|{\mathscr P}_tf(x)-c_f e^{-\lambda_L t}\pi_L(x) \right|
\leq c \|f\|_{L^\infty}e^{-(\lambda_L+\alpha) t}
\eel{ltpricing}
for all times $t\geq t_0$ and each $x\in E$, where $\|f\|_{L^\infty}=\sup_{x\in E}|f(x)|$.
\end{theorem}
The exponential ergodicity estimate \eqref{exp_ergo} ensures that the distribution of $X_t$ under $\mathbb{Q}^{\pi_R}$ converges to the limiting distribution $\varsigma$ sufficiently fast so that the expectation $\mathbb{E}^{\mathbb{Q}^{\pi_R}}_x\left[f(X_t)/\pi_R(X_t)\right]$ for any bounded function $f$ converges to the expectation under the limiting distribution at the exponential rate $\alpha>0$.

Exponential ergodicity is sufficient to identify $\pi_R$ with $\pi_L$. It is also sufficient for the conditions in Theorem \ref{asym_long}, and thus the results in Theorem \ref{implication_L1} and Theorem \ref{asym_long} hold. The long-term pricing formula \eqref{ltpricing} complements the long term pricing formula in Proposition 7.1 of \citet{hansen_2009} (see also Section 4 of \citet{borovicka_2014mis}) by also providing the exponential convergence rate to the long term limit. Indeed, Eq.\eqref{ltpricing} is obtained from Eq.\eqref{exp_ergo} via Eq.\eqref{pricingopqpi}.

When Assumption \ref{exp_ergo_assumption} is not satisfied, it may be the case that $\mathbb{Q}^{\pi_R}$ exists while $\mathbb{L}$ does not, or $\mathbb{L}$ exists while $\mathbb{Q}^{\pi_R}$ does not, or that both $\mathbb{L}$ and $\mathbb{Q}^{\pi_R}$ exist but are distinct. In order to identify $\mathbb{Q}^{\pi_R}$ with $\mathbb{L}$ by invoking stochastic stability assumptions weaker than Assumption 5.2, we need to ensure that the constant payoff (of the zero coupon bond) is within the scope of the long term pricing formula, i.e. $\int_E \frac{1}{\pi_R(y)}\varsigma(dy)<\infty$.
If this condition fails, the long forward measure fails to exist even though the long term pricing formula holds for a class of payoffs that does not include constants.
 \citet{borovicka_2014mis} explicitly impose this integrability condition in their Section 4, thus ruling out such cases.
Nevertheless, even under this integrability condition, we  can generally only show almost sure convergence of $B_t^T$ to $B_t^\infty$ for each $t$, which does not imply the ucp convergence of $B_t^T$ and the strong convergence of $T$-forward measures $\mathbb{Q}^T$ to ${\mathbb L}$. The exponential ergodicity Assumption 5.2 ensures the $L^1$ convergence condition of Theorem \ref{implication_L1} and, thus, ucp and stronger semimartingale convergence of processes $B_t^T$ to the long bond and convergence in total variation of $\mathbb{Q}^T$ to ${\mathbb L}$.


\section{Concluding Remarks}

This paper provided a unified treatment of the
long-term factorization of the pricing kernel of \citet{alvarez_2005using}, \citet{hansen_2009} and \citet{hansen_2012} in the semimartingale setting without Markovian assumptions. The transitory component discounts at the stochastic rate of return on the long bond and is further factorized into discounting at the asymptotic long bond yield $\lambda$ and a positive semimartingale that, in the Markovian setting, is expressed in terms of a positive eigenfunction of the pricing operator with the eigenvalue $e^{-\lambda t}$.
The permanent component is a martingale that accomplishes a change of probabilities to the long forward measure, the limit of $T$-forward measures.
The long forward measure is further interpreted as the long-term risk-neutral measure, as the long-term risk premia for stochastically growing cash flows vanish under ${\mathbb L}$.
The approach of this paper via semimartingale convergence to the long-term limit extends and complements the operator approach of \citet{hansen_2009} and unifies it with the approach of \citet{alvarez_2005using}.

The volatility of the permanent component drives the wedge between data-generating and long forward probabilities.
 \citet{linetsky2016bond} apply this methodology to the empirical analysis of the US Treasury bond market and show that the martingale component is highly volatile and controls the term structure of the risk-return trade-off in the bond market. These results on the economic significance of the martingale component complement the non-parametric  results of \citet{alvarez_2005using} and \citet{bakshi_2012} based on bounds and of \citet{borovicka_2014mis} and \citet{christensen2014nonparametric} based on structural asset pricing models calibrated to macro-economic fundamentals and impose significant economic restrictions on asset pricing models.

\bibliography{mybib7}

\newpage
\setcounter{page}{1}
\appendix
\begin{center}
{\Large Supplementary Appendix to\\
Long Term Risk: A Martingale Approach}\\
\enskip\\
Likuan Qin and Vadim Linetsky
\end{center}

\section{Proofs for Section \ref{long_term_section}}
\label{proof_l1}
We first recall some results about semimartingale topology originally introduced by \citeappendix{emery_1979topologie} (see \citeappendix{czichowsky_2011closedness}, \citeappendix{kardaras_2013closure} and \citeappendix{cuchiero_2014convergence} for recent applications in mathematical finance).
The semimartingale topology is stronger than the topology of uniform convergence in probability on compacts (ucp). In the latter case the supremum in Eq.\eqref{emeryd} is only taken over integrands in the form $\eta_t=1_{[0,s]}(t)$ for every $s> 0$:
\be
d_{\text{ucp}}(X,Y)=\sum_{n\geq 1} 2^{-n} {\mathbb E}^\mathbb{P}[1 \wedge \sup_{s\leq n}|X_s-Y_s| ].
\eel{ucpd}

The following inequality due to Burkholder is useful for proving convergence in the semimartingale topology in Theorem \ref{implication_L1} (see \citeappendix{meyer_1972martingales} Theorem 47, p.50 for discrete martingales and \citeappendix{cuchiero_2014convergence} for continuous martingales, where a proof is provided inside the proof of their Lemma 4.7).
\begin{lemma}
\label{meyer_S}
For every martingale $X$ and every
predictable process $\eta$ bounded by one, $|\eta_t|\leq 1$, it holds that
\be
a{\mathbb P}\left(\sup_{s\in[0,t]} \left| \int_0^s \eta_u dX_u \right| > a\right)\leq 18 {\mathbb E}^\mathbb{P}[|X_t|]
\eel{Burk}
for all $a\geq 0$ and $t>0$.
\end{lemma}
We will also use the following result (see \citeappendix{kardaras_2013closure}, Proposition 2.10).
\begin{lemma}
\label{product_S}
If $X^n\xrightarrow{\rm {\cal S}} X$ and $Y^n\xrightarrow{\rm {\cal S}} Y$, then $X^n Y^n\xrightarrow{\rm {\cal S}} XY$.
\end{lemma}
We will also make use of the following lemma.
\begin{lemma}
\label{L1_martingale}
Let $(X_t^n)_{t\geq 0}$ be a sequence of martingales such that   $X_t^n\xrightarrow{\rm L^1}X_t^\infty$ for each $t\geq 0$. Then $(X_t^\infty)_{t\geq 0}$ is a martingale.
\end{lemma}
\begin{proof}
It is immediate that ${\mathbb E}[|X_t^\infty|]<\infty$ for all $t$. We need to verify that ${\mathbb E}_s[X_t^\infty]=X_s^\infty$ for $t>s\geq 0$.
First we show that from $X_t^n\xrightarrow{\rm L^1}X_t^\infty$ it follows that
\be \mathbb{E}_s[X_t^n]\xrightarrow{\rm L^1}\mathbb{E}_s[X_t^\infty] \eel{l1lemma}
for each
$s<t$. By Jensen's inequality, for each $0\leq s<t$ we have
$\left|\mathbb{E}_s[X_t^{n}-X_t^{\infty}]\right|\leq \mathbb{E}_s[|X_t^{n}-X_t^{\infty}|].$ Taking expectations on both sides, we have
\be
\mathbb{E}\left|\mathbb{E}_s[X_t^n]-\mathbb{E}_s[X_t^\infty]\right|\leq\mathbb{E}_s[|X_t^{n}-X_t^{\infty}|].
\ee
Thus, $X_t^n\xrightarrow{\rm L^1} X_t^\infty$ implies $\mathbb{E}_s[X_t^n]\xrightarrow{\rm L^1} \mathbb{E}_s[X_t^\infty]$ for each $s<t$.

Since $X_t^n$ are martingale, $\mathbb{E}_s[X_t^n]=X_s^n$.
By \eqref{l1lemma} for $t\geq s$, $X_s^n\xrightarrow{\rm L^1}\mathbb{E}_s[X_t^\infty]$.
On the other hand, $X_s^n\xrightarrow{\rm L^1}X_s^\infty$. Thus, $\mathbb{E}_s[X_t^\infty]=X_s^\infty$ for $t>s$, hence $X_t^\infty$ is a martingale.
\end{proof}

\noindent{\emph{Proof of Theorem \ref{implication_L1}}. (i) It is easy to see Eq.\eqref{PKL1} implies $M_t^T$ converges to $M_t^\infty$ in $L^1$ under $\mathbb{P}$. Since $(M_t^T)_{t\geq 0}$ are positive $\mathbb{P}$-martingales with $M_0^T=1$, and for each $t\geq 0$ random variables $M_t^T$ converge to $M_t^\infty>0$ in $L^1$, by Lemma \ref{L1_martingale} $(M_t^\infty)_{t\geq 0}$ is also a positive $\mathbb{P}$-martingale with $M_0^\infty=1$.
Emery's distance between the martingale $M^T$ for some $T>0$ and $M^\infty$ is
\be
d_{\cal S}(M^T,M^\infty)=\sum_{n\geq 1} 2^{-n} \sup_{|\eta|\leq 1}{\mathbb E}^\mathbb{P}\left[1 \wedge \left|\int_0^n \eta_s d(M^T-M^\infty)_s\right| \right].
\ee
To prove $M^T\xrightarrow{\rm {\cal S}} M^\infty$, it suffices to prove that for all $n$
\be
\lim_{T\rightarrow\infty} \sup_{|\eta|\leq 1}{\mathbb E}^\mathbb{P}\left[1 \wedge \left|\int_0^n \eta_s d(M^T-M^\infty)_s\right| \right]=0.
\eel{M_S_P}
We can write for an arbitrary $\epsilon>0$ (for any random variable $X$ it holds that ${\mathbb E}[1\wedge |X|]\leq {\mathbb P}(|X|>\epsilon)+\epsilon$)
\be
{\mathbb E}^\mathbb{P}\left[1 \wedge \left|\int_0^n \eta_s d(M^T-M^\infty)_s\right| \right]\leq \mathbb{P}\left(\left|\int_0^n\eta_s d(M^T-M^\infty)_s\right|>\epsilon\right)+\epsilon.
\ee
By Lemma \ref{meyer_S},
\be
\sup_{|\eta|\leq1}{\mathbb E}^\mathbb{P}\left[1 \wedge \left|\int_0^n \eta_s d(M^T-M^\infty)_s\right| \right]\leq \frac{18}{\epsilon}\mathbb{E}^\mathbb{P}[|M_n^T-M_n^\infty|]+\epsilon.
\ee
Since  $\lim_{T\rightarrow\infty}\mathbb{E}^\mathbb{P}[|M_n^T-M_n^\infty|]=0$, and $\epsilon$ can be taken arbitrarily small, Eq.\eqref{M_S_P} is verified and, hence, $M^T\xrightarrow{\rm {\cal S}} M^\infty$.


\noindent (ii) We have shown that $S_t B_t^T=M_t^T\xrightarrow{\rm {\cal S}}M_t^\infty:=S_t B_t^\infty$. By Lemma \ref{product_S}, $B_t^T\xrightarrow{\rm {\cal S}}B_t^\infty$, and $B_t^\infty$ is the long bond according to Definition \ref{def_longbond} (the semimartingale convergence is stronger than the ucp convergence).

\noindent Part (iii) is a direct consequence of (i) and (ii).

\noindent (iv) Define a new probability measure ${\mathbb Q}^\infty$ by $\mathbb{Q}^\infty|_{\mathscr{F}_t}=M_t^\infty\mathbb{P}|_{\mathscr{F}_t}$ for each $t\geq 0$. The distance in total variation between the measure ${\mathbb Q}^T$ for some $T>0$ and ${\mathbb Q}^\infty$ on $\mathscr{F}_t$ is:
$$2 \sup_{A\in\mathscr{F}_t}|{\mathbb{Q}^T}(A)-\mathbb{Q}^\infty(A)|.$$
For each $t\geq 0$ we can write:
\be
0=\lim_{T\rightarrow\infty}\mathbb{E}^\mathbb{P}[|M_t^T-M_t^\infty|]=   \lim_{T\rightarrow\infty}\mathbb{E}^\mathbb{P}[M_t^\infty |B_t^T/B_t^\infty-1|]   = \lim_{T\rightarrow\infty}\mathbb{E}^{\mathbb{Q}^\infty}[|B_t^T/B_t^\infty-1|].
\ee
Thus,
\be
\lim_{T\rightarrow\infty}\sup_{A\in\mathscr{F}_t}\left|\mathbb{E}^{\mathbb{Q}^\infty}\left[(B_t^T/B_t^\infty){\bf 1}_A\right]-\mathbb{E}^{\mathbb{Q}^\infty}\left[{\bf 1}_A\right]\right|=0.
\ee
Since
$\left.\frac{d\mathbb{Q}^T}{d\mathbb{Q}^\infty}\right|_{\mathscr{F}_t}=\frac{B_t^T}{B_t^\infty},$
it follows that
\be
\lim_{T\rightarrow\infty}\sup_{A\in\mathscr{F}_t}\Big|\mathbb{E}^{\mathbb{Q}^T}[{\bf 1}_A]-\mathbb{E}^{\mathbb{Q}^\infty}[{\bf 1}_A]\Big|=0.
\ee
Thus, $\mathbb{Q}^T$ converge to $\mathbb{Q}^\infty$ in total variation on $\mathscr{F}_t$ for each $t$. Since convergence in total variation implies strong convergence of measures, this shows that $\mathbb{Q}^\infty$ is the long forward measure according to Definition \ref{def_longforward}, $\mathbb{Q}^\infty=\mathbb{L}$. $\Box$\\
%
%

\noindent{\emph{Proof of Theorem \ref{asym_long}}.
(i) Define functions $h(t):=P_0^{\log t}$ and $g(t):=\lim_{T\rightarrow \infty}P_0^{T-t}/P_0^T$ (the latter is defined for each $t$ due to our assumption).  Then for all $0<a<1$
\be
\lim_{t\rightarrow\infty}\frac{h(at)}{h(t)}=\lim_{t\rightarrow\infty}\frac{P_0^{\log at}}{P_0^{\log t}}=\lim_{t\rightarrow\infty}\frac{P_0^{\log t + \log a}}{P_0^{\log t}}=g(-\log a).
\ee
Thus, $h(t)$ is a regularly varying function (see \citeappendix{bingham1989regular}). By Karamata's characterization theorem (see \citeappendix{bingham1989regular} Theorem 1.4.1), there exists a real number $\lambda$ such that $\lim_{t\rightarrow\infty}\frac{h(at)}{h(t)}=g(-\log a)=a^{-\lambda}$ and a slowly varying function $L(t)$ such that $h(t)=t^{-\lambda}L(t)$. Rewriting it gives $g(t)=e^{\lambda t}$ and $P_0^t=e^{-\lambda t}L(e^t)$.\\
(ii) By Eq.\eqref{PKL1}, $S_tP_t^T/P_0^T$ converges to $M_t^\infty$ in $L^1$ under $\mathbb{P}$ as $T\rightarrow \infty$. Thus it also converges in probability under $\mathbb{P}$, as well as under any measure locally equivalent to $\mathbb{P}$. Hence $P_t^T/P_0^T$, as well as  $\log(P_t^T/P_0^T)$, converge in probability (from now on, as well as in the proof of Theorems 3.3-3.5, we omit explicit dependency on the probability measure when we talk about convergence in probability, since it holds under all locally equivalent measures). Thus, $\log(P_t^T/P_0^T)/(T-t)$ converges to zero in probability.
Since $P_0^T=e^{-\lambda T}L(e^T)$ and for any slowly varying function $\lim_{T\rightarrow\infty}\frac{1}{T-t}\log(L(e^T))=0$ (see Proposition 1.3.6 of \citeappendix{bingham1989regular}), we have $$\lim_{T\rightarrow\infty}\log P_0^T/(T-t)=-\lambda.$$
Combining these two property yields (ii).\\
(iii) and (iv) By \citeappendix{bingham1989regular} Theorem 1.2.1, $P_0^T/P_0^{T-t}$ converges to $1/g(t)$ as $T\rightarrow \infty$ uniformly on compacts, and thus also in semimartingale topology. By Theorem \ref{implication_L1}, $B_t^T$ (and thus $P_t^T/P_0^T$) converges to $B_t^\infty$ in semimartingale toplogy. Thus by Lemma \ref{product_S}, the ratio $P_t^T/P_0^{T-t}$ converges in semimartingale topology as $T\rightarrow\infty$, and we denote the limit $\pi_t$. The decomposition of $B_t^\infty$ is then immediate.\\
(v) Since $M_t^\infty=S_tB_t^\infty=S_t\pi_t e^{\lambda t}$ is a martingale, we have $\mathbb{E}_t^\mathbb{P}[S_T\pi_T e^{\lambda T}]=S_t \pi_t e^{\lambda t}$. Re-writing it yields Eq.\eqref{eigens}. Combining the fact that $P_t^T=\mathbb{E}_t^\mathbb{L}[B_t^\infty/B_T^\infty]=e^{-\lambda (T-t)}\mathbb{E}_t^\mathbb{L}[\pi_t/\pi_T]$ and Eq.\eqref{lt_rate} yields Eq.\eqref{Lmean}.
$\Box$\\

\noindent{\emph{Proof of Theorem \ref{long_holding}}.
(i) By assumption, we have for $T>T'$
\be
\frac{c}{C\mathbb{E}_t^\mathbb{L}[1/\pi_T]}<\frac{\mathbb{E}_t^\mathbb{L}[C_T]}{\mathbb{E}_t^\mathbb{L}[C_T/\pi_T]}<\frac{C}{c\mathbb{E}_t^\mathbb{L}[1/\pi_T]}.
\ee
Combining it with Eq.\eqref{Lmean} yields that $\log\left(\frac{\mathbb{E}_t^\mathbb{L}[C_T]}{\mathbb{E}_t^\mathbb{L}[C_T/\pi_T]}\right)/(T-t)$ converges to zero in probability. Substituting it into Eq.\eqref{rho_exp}, we arrive at part (i). Part (ii) is proved similarly. $\Box$\\

\noindent{\emph{Proof of Theorem \ref{long_power_yield}}.
(i) Since $P_0^T=O(t^{-\gamma})$, $\lambda=0$ and $B_t^\infty=\pi_t$. Similar to the proof of Theorem \ref{asym_long}, $(-\log P_t^T)/ \log(T-t)$ converges to $\gamma$ in probability as $T\rightarrow\infty$. Since $P_t^T=\mathbb{E}_t^\mathbb{L}[B_t^\infty/B_T^\infty]=\mathbb{E}_t^\mathbb{L}[\pi_t/\pi_T]$, we have $(-\log\mathbb{E}_t^\mathbb{L}[1/\pi_T])/\log(T-t)$ converges to $\gamma$ in probability. By assumption, we have for $T>T'$
\be
\frac{c}{C\mathbb{E}_t^\mathbb{L}[1/\pi_T]}<\frac{\mathbb{E}_t^\mathbb{L}[C_T]}{\mathbb{E}_t^\mathbb{L}[C_T/\pi_T]}<\frac{C}{c\mathbb{E}_t^\mathbb{L}[1/\pi_T]}.
\ee
Thus $\log\left(\frac{\mathbb{E}_t^\mathbb{L}[C_T]}{\mathbb{E}_t^\mathbb{L}[C_T/\pi_T]}\right)/\log(T-t)$ converges to $\gamma$ in probability. (ii) can be proved similarly. $\Box$\\

\noindent{\emph{Proof of Theorem \ref{long_growth_yield}}.
By assumptions and Eq.\eqref{growth_yield} we can write by changing the probability measure to ${\mathbb G}$:
\be
\rho_{t,T}^\mathbb{L}(G_T)=\lambda+\frac{1}{T-t}\log\left(\frac{\mathbb{E}_t^{\mathbb{G}}[\pi_T/\pi_T^G]}{\pi_t\mathbb{E}_t^{\mathbb{G}}[1/\pi_T^G]}\right).
\ee
By assumption, we immediately have Eq.\eqref{growth_limit_yield}. $\Box$

\section{Discrete Time Environment}
\label{app_aj}

We will show how the results of  \citeappendix{alvarez_2005using} in discrete-time environments are naturally nested in our Theorems \ref{implication_L1} and \ref{asym_long}.
\citeappendix{alvarez_2005using} work in discrete time with the pricing kernel $S_t,$ $t=0,1,...,$ and make the following assumptions (below $P_t^{t+\tau}$ is the time-$t$ price of a pure discount bond with maturity at time $t+\tau$ and unit face value, where $t,\tau=0,1,...$).
\begin{assumption} {\bf (\citeappendix{alvarez_2005using} Assumptions 1 and 2)}
\label{assumption_aj}\\
(i) There exists a constant $\lambda$ such that
$0<\lim_{\tau\rightarrow\infty}e^{\lambda \tau}P_t^{t+\tau}<\infty$
almost surely for all $t=0,1,...$.\\
(ii) For each $t=1,2,...$ there exists a random variable $x_t$ with ${\mathbb E}^{\mathbb P}_{t-1}[x_t]<\infty$ such that
$e^{\lambda (t+\tau)}S_tP_t^{t+\tau}\leq x_t$
almost surely for all $\tau=0,1,...$.
\end{assumption}

Any discrete-time adapted process can be embedded into a continuous-time semimartingale as follows. For a discrete-time process $(X_t,t=0,1,...)$,
define a continuous-time process $(\tilde{X}_t)_{t\geq 0}$ such that at integer times it takes the same values as the discrete time process $X$, and is piece-wise constant between integer times, i.e. $\tilde{X}_t=X_{[t]}$, where $[t]$ denotes the integer part (floor) of $t$.
This process has RCLL paths and is a semimartingale (it is of finite variation).
The following result shows that Proposition 1 in \citeappendix{alvarez_2005using} is nested in Theorem \ref{implication_L1} and \ref{asym_long}.

\begin{proposition}
\label{aj_l1}
Consider a discrete-time positive pricing kernel $(S_t,t=0,1,\ldots)$ with ${\mathbb E}^{\mathbb P}[S_{t+\tau}/S_t]<\infty$ for all $t,\tau$. Suppose pure discount bonds $P_t^{t+\tau}= {\mathbb E}^{\mathbb P}_t[S_{t+\tau}/S_t]$ satisfy
Assumption \ref{assumption_aj}. Then the corresponding continuous-time positive semimartingale pricing kernel $(\tilde{S}_t)_{t\geq 0}$
satisfies the conditions in Theorem \ref{implication_L1} and Theorem \ref{asym_long}, hence, all results in Theorem \ref{implication_L1} and Theorem \ref{asym_long} hold.
\end{proposition}

\noindent{\emph{Proof of Proposition \ref{aj_l1}}. We first prove Eq.\eqref{PKL1} is satisfied. We first consider integer values of $t$ and $\tau$.
By assumption (ii), we have $$S_t\frac{P_t^{t+\tau}}{P_0^{t+\tau}}e^{\lambda ({t+\tau})}P_0^{t+\tau}\leq x_t.$$
Recall that $$M_t^{t+\tau}=\frac{\mathbb{E}_t^\mathbb{P}[S_{t+\tau}]}{\mathbb{E}^\mathbb{P}[S_{t+\tau}]}=S_t\frac{P_t^{t+\tau}}{P_0^{t+\tau}}$$ for $t,\tau\geq0$. Thus, we have that
$$
M_t^{t+\tau}e^{\lambda ({t+\tau})}P_0^{t+\tau}\leq x_t
$$
for $t,\tau\geq0$.
By assumption (i), $e^{\lambda T}P_0^{T}$ has a positive finite limit as $T\rightarrow \infty$. Thus, there exists a constant $c>0$ such that $e^{\lambda T}P_0^T>c$ for all $T$. Hence, $M_t^{t+\tau}\leq c^{-1}x_t$ for $t,\tau\geq0$.

Furthermore, for all $t,\tau\geq0$ we can write
$$
M_t^{t+\tau}=S_t\frac{P_t^{t+\tau}}{P_0^{t+\tau}}=e^{\lambda t}S_t\frac{e^{\lambda\tau}P_t^{t+\tau}}{e^{\lambda (t+\tau)}P_0^{t+\tau}}.
$$
By assumption (i), $\lim_{\tau\rightarrow\infty}M_t^{t+\tau}$ exists almost surely and is positive. We denote it $M_t^\infty$.
Since $M_t^{t+\tau}\leq c^{-1}x_t$ and $x_t$ is integrable, by the Dominated Convergence Theorem we have that
$M_t^{t+\tau}\rightarrow M_t^\infty$
in $L^1$ as $\tau\rightarrow \infty$ for each fixed integer $t=0,1,...$.

We now consider real values of $t$ and $\tau$ and recall our embedding of discrete-time adapted processes into continuous semimartingales with piece-wise constant paths.
For each real $t$ and $\tau$ we have $M_t^{t+\tau}=M_n^N$, where $n,N$ are two integers such that $t\in[n,n+1)$ and  $t+\tau\in[N,N+1)$. Thus, $M_t^{t+\tau}\rightarrow M_t^\infty$
in $L^1$ as $\tau\rightarrow \infty$ for each fixed real $t\geq 0$. This prove Eq.\eqref{PKL1}.

$P_0^{T-t}/P_0^T$ converges for all $t>0$ is a simple consequence of the fact that $e^{\lambda t}P_0^t$ converges for all $t>0$.
$\Box$

\section{Proofs for Section \ref{example_Markov}}
\label{appendix_markov}


We start with proving the following measurability property of the bond pricing function $P(t,x)$ under Assumption 5.1.
\begin{lemma}
\label{markov_measurable}
If the pricing kernel $S_t$ satisfies Assumption 5.1, then the bond pricing function $P(t,x)={\mathbb E}_x^{\mathbb P}[S_t]$ is jointly measurable with respect to ${\cal B}(\mathbb{R}_+)\otimes{\cal E}^*$, where ${\cal B}(\mathbb{R}_+)$ is the Borel $\sigma$-algebra on $\mathbb{R}_+$, ${\cal E}^*$ is the $\sigma$-algebra of universally measurable sets on $E$ (see \citeappendix{sharpe_1988} p.1).
\end{lemma}
\begin{proof}
Let $P^n(t,x)=\mathbb{E}^\mathbb{P}_x[S_t\wedge n]$.
By \citeappendix{chen_2011} Exercise A.1.20 for fixed $t$ $P^n(t,x)$ is ${\cal E}^*$-measurable.
Since $S_t$ is right continuous, by the Bounded Convergence Theorem for fixed $x$ the function $P^n(t,x)$ is right continuous in $t$. Thus, on $[0,1)\times E$ we can write:
\be
P^n(t,x)=\lim_{m\rightarrow\infty}P^n_m(t,x),
\ee
where
\be
P^n_m(t,x):=\sum_{i=1}^m 1_{[(i-1)/m,i/m)}(t) P^n((i-1)/m,x).
\eel{fnm_fn}
Thus, on $[0,1)\times E$ the function $P^n_m(t,x)$ is jointly measurable with respect to ${\cal B}([0,1))\otimes{\cal E}^*$. Similarly we can prove that $P^n_m(t,x)$ is jointly measurable with respect to ${\cal B}(\mathbb{R}_+)\otimes{\cal E}^*$. By Eq.\eqref{fnm_fn}, $P^n(t,x)$ is then also jointly measurable with respect to ${\cal B}(\mathbb{R}_+)\otimes{\cal E}^*$. Since $S_t$ is integrable, by the Dominated Convergence Theorem $\lim_{n\rightarrow\infty} P^n(t,x)=P(t,x)$. Thus, $P(t,x)$ is also jointly measurable with respect to ${\cal B}(\mathbb{R}_+)\otimes{\cal E}^*$.
\end{proof}
Next we prove the following result.
\begin{lemma} Suppose the PK $S$ satisfies Assumption 5.1 and Eq.(3.1) holds under ${\mathbb P}_x$ for each $x\in E$.
Then for each $t>0$ and $x\in E$ we can write for the long bond
\be
B_t^\infty(x)=b^\infty(t,x,X_t)>0
\eel{LBfunction}
${\mathbb P}_x$-almost surely,
where $b^\infty(t,x,y)$ is a universally measurable function of $y$ for each fixed $t>0$ and $x\in E$.
\end{lemma}
{\em Proof.}
The long bond $B_t^\infty(x)$ is the ucp limit of the processes  $B_t^T(x)$ defined in Section 3.
Dependence on the initial state $X_0=x$ comes from dividing by the initial bond price $P(0,x)$ at time zero in the definition of $B_t^T$.
For each $t>0$ and $x\in E$, the random variables $B_t^T(x)=P_t^T(x)/P_0^T(x)=P(T-t,X_t)/P(T,x)$ with $T\geq t$ converge to $B_t^\infty(x)$ as $T\rightarrow \infty$ in probability. By Lemma \ref{markov_measurable}, $P(T-t,X_t)/P(T,x)$ is $\sigma(X_t)$-measurable ($X_t$ is viewed as a random element taking values in $E$ equipped with the $\sigma$-algebra ${\cal E}^*$, thus $\sigma(X_t)$ is generated by inverses of universally measurable sets). Its limit in probability $B_t^\infty(x)$ can also be taken $\sigma(X_t)$-measurable and, by Doob-Dynkin lemma, we can write it as $b^\infty(t,x,X_t)$, where for each fixed $t>0$ and $x\in E$, $b^\infty(t,x,y)$ is a universally measurable function of $y$. $\Box$\\
\\
By Lemma, for each $t>0$ and $x\in E$ the random variables $P(T-t,X_t)/P(T,x)$ converge to the random variable $b^\infty(t,x,X_t)$ in probability under $\mathbb{P}_x$.
In Theorem 5.1, we strengthen it to pointwise convergence of the function $P(T-t,y)/P(T,x)$ as $T$ goes to infinity, i.e. for each $t>0$ and $x,y\in E$:
\be
\lim_{T\rightarrow\infty}\frac{P(T-t,y)}{P(T,x)}=b^\infty(t,x,y)>0.
\eel{f_converge_g}
Now we are ready to prove Theorem \ref{L_equal_eigen}.

{\em Proof of Theorem \ref{L_equal_eigen}.}
By Lemma \ref{markov_measurable} $P(t,x)$ is jointly measurable with respect to ${\cal B}(\mathbb{R}_+)\otimes{\cal E}^*$. Thus, by Eq.\eqref{f_converge_g}, $P(T-t,y)/P(T,x)$ is jointly measurable with respect to ${\cal B}(\mathbb{R}_+)\otimes{\cal E}^*\otimes{\cal E}^*$. Thus, the function $b^\infty(t,x,y)$ is also jointly measurable with respect to ${\cal B}(\mathbb{R}_+)\otimes{\cal E}^*\otimes{\cal E}^*$.

For any $t,s>0$ and $x,y,z\in E$ we can write:
\be
b^\infty(t+s,y,z)=\lim_{T\rightarrow \infty} \frac{P(T-t,z)}{P(T+s,y)}=\lim_{T\rightarrow \infty} \frac{P(T,x)}{P(T+s,y)}\frac{P(T-t,z)}{P(T,x)}
\ee
\be
=b^\infty(s,y,x)b^\infty(t,x,z).
\eel{gtxy}
Taking $x=y=z$ in Eq.\eqref{gtxy}, we have
$$b^\infty(t,x,x)b^\infty(s,x,x)=b^\infty(t+s,x,x),$$
which implies that for each fixed $x\in E$ $b^\infty(t,x,x)$ satisfies Cauchy's multiplicative functional equation as a function of time.
Since $b^\infty(t,x,y)$ is jointly measurable with respect to ${\cal B}(\mathbb{R}_+)\otimes{\cal E}^*\otimes{\cal E}^*$, for fixed $x$ $\ln b^\infty(t,x,x)$ is measurable with respect to ${\cal B}(\mathbb{R}_+)$. It is known that a Borel measurable function that satisfies Cauchy's functional equation is linear. Thus, we have that $b^\infty(t,x,x)=e^{\lambda_L(x) t}$.

Again by Eq.\eqref{gtxy}, for any $x,y\in E$ we have
$$b^\infty(2t,y,x)=b^\infty(t,y,x)b^\infty(t,x,x)=b^\infty(t,y,y)b^\infty(t,y,x),$$
and we have $b^\infty(t,y,y)=b^\infty(t,x,x)$. Thus, $\lambda_L(x)$ is independent of $x$.
Taking $y=x$ in Eq.\eqref{gtxy}, we have $b^\infty(t+s,x,z)=e^{\lambda_L s}b^\infty(t,x,z)$. Thus, $e^{-\lambda_L t}b^\infty(t,x,z)$ is independent of $t$.
Fix $x_0\in E$ and define $\pi_L(x):=e^{-\lambda_L t}b^\infty(t,x_0,x)$. It is independent of $t$ and $x_0$ is fixed. By Eq.\eqref{gtxy}, $b^\infty(t,x_0,x)b^\infty(t,x,x_0)=b^\infty(2t,x_0,x_0)=e^{-2\lambda_L t}$. Thus, $b^\infty(t,x,x_0)=e^{\lambda_L t}1/ \pi_L(x)$.
Finally, we have
$$b^\infty(t,x,y)=b^\infty(t/2,x,x_0)b^\infty(t/2,x_0,y)=e^{\lambda_L t}\frac{\pi_L(y)}{\pi_L(x)}.$$
By Eq.\eqref{LBfunction} we then have
$$B_t^\infty(x)=e^{\lambda_L t}\frac{\pi(X_t)}{\pi(x)}.$$
$\pi_L$ is an eigenfunction of the pricing operators ${\mathscr P}_t$ with the eigenvalues $e^{-\lambda_L t}$ from the fact that $M_t^\infty=S_tB_t^\infty$ is a martingale. Thus, we arrive at the identification of the long forward measure with an eigen-measure associated with the eigenfunction $\pi_L$, and the identification $\mathbb{L}=\mathbb{Q}^{\pi_L}$ thus follows.
$\Box$
\begin{remark}
We note the difference between the setting here and the one in \citeappendix{linetsky_2014_cont}. Here we do not assume that the pricing operator maps Borel functions to Borel functions upfront. Since the long bond $e^{\lambda_L t}\frac{\pi_L(X_t)}{\pi_L(x)}$ is a right continuous semimartingale, by \citeappendix{cinlar_1980} the function $\pi_L$ is locally the difference of two 1-excessive functions. For a Borel right process, its excessive functions are generally only universally measurable, but not necessarily Borel measurable. Thus the eigenfunction $\pi_L$ we find above is also not necessarily Borel measurable, but is universally measurable.
Hence after the measure change from the data generating measure to the long forward measure, under $\mathbb{L}=\mathbb{Q}^{\pi}$ the Markov process $X$ may not be a Borel right process, but it is a right process. If we explicitly assume that the pricing operator maps Borel functions to Borel functions, as is done in  \citeappendix{linetsky_2014_cont}, then the eigenfunction $\pi_L$ is automatically Borel and $X$ is a Borel right process under $\mathbb{Q}^{{\pi}_L}$.
Here we opted for this slightly more general set up, so not to impose further restrictions on the pricing kernel.
\end{remark}

\noindent\emph{Proof of Theorem \ref{markov_long}}. Let $Q^{\pi_R}(t,x,\cdot)$ denote the transition measure of $X$ under $\mathbb{Q}^{\pi_R}$.
We verify the $L^1$ convergence condition Eq.\eqref{PKL1} with $M_t^\infty=M_t^{\pi_R}$ with the martingale associated with the recurrent eigenfunction. This then identifies $e^{\lambda_R t}\frac{\pi_R(X_t)}{\pi_R(X_0)}$ with the long bond $B_t^{\infty}$  and the recurrent eigen-measure with the long forward measure, $\mathbb{Q}^{\pi_R}=\mathbb{L}$.

We note that for  any valuation process $V$, the condition \eqref{PKL1} can be written under any locally equivalent probability measure ${\mathbb Q}^V$ defined by $\mathbb{Q}^V|_{\mathscr{F}_t}=S_tR_{0,t}^V \mathbb{P}|_{\mathscr{F}_t}$:
\be
\lim_{T\rightarrow\infty}\mathbb{E}^{\mathbb{Q}^V}[|B_t^T/V_t-B_t^\infty/V_t|]=0.
\eel{L1_alter}
We can use this freedom to choose the measure convenient for the setting at hand. Here we choose to verify it under $\mathbb{Q}^{\pi_R}$, i.e. we will verify Eq.\eqref{L1_alter} under $\mathbb{Q}^V=\mathbb{Q}^{\pi_R}$ due to its convenient form.
Since
$$P_t^T =e^{-\lambda_R(T-t)}\pi_R(X_t)\mathbb{E}^{\mathbb{Q}^{\pi_R}}_{X_t}[\frac{1}{\pi_R(X_{T-t})}],$$
we have
\be
e^{-\lambda_R t}\frac{P_t^T \pi_R(X_0)}{P_0^T \pi_R(X_t)}=\frac{\mathbb{E}^{\mathbb{Q}^{\pi_R}}_{X_t}[\frac{1}{\pi_R(X_{T-t})}]}{\mathbb{E}^{\mathbb{Q}^{\pi_R}}_{X_0}[\frac{1}{\pi_R(X_T)}]}.
\eel{expression_ratio}

Let $J:=\int_E \varsigma(dy)\frac{1}{\pi_R(y)}$ (it is finite by Assumption \ref{exp_ergo_assumption}). Since
$$\mathbb{E}_{x}^{\mathbb{Q}^{\pi_R}}[\frac{1}{\pi_R(X_{t})}]=\int_E Q^{\pi_R}(t,x,dy)\frac{1}{\pi_R(y)},$$
by Eq.\eqref{exp_ergo} we have for $T-t\geq t_0$:
\be
J-\frac{c}{\pi_R(X_t)}e^{-\alpha(T-t)}\leq\mathbb{E}_{X_t}^{\mathbb{Q}^{\pi_R}}[\frac{1}{\pi_R(X_{T-t})}]\leq J+\frac{c}{\pi_R(X_t)}e^{-\alpha(T-t)},
\eel{markov_bound1}
and for each initial state $X_0=x\in E$ and $T\geq\max(T_0,t+t_0)$:
\be
J-\frac{c}{\pi_R(x)}e^{-\alpha T}\leq\mathbb{E}_{x}^{\mathbb{Q}^{\pi_R}}[\frac{1}{\pi_R(X_T)}]\leq J+\frac{c}{\pi_R(x)}e^{-\alpha T}.
\eel{markov_bound2}
For each $x\in E$ there exists $T_0$ such that for $T\geq T_0$,
$\frac{c}{\pi_R(x)}e^{-\alpha T}\leq J/2$.
We can thus write for each $x\in E$:
\be
-1\leq
e^{-\lambda_R t}\frac{P_t^T \pi_R(x)}{P_0^T \pi_R(X_t)}-1
\leq \frac{2}{J} \left(\frac{c}{\pi_R(X_t)}e^{-\alpha(T-t)}+\frac{c}{\pi_R(x)}e^{-\alpha T}\right),
\ee
Thus,
\be
\left|e^{-\lambda_R t}\frac{P_t^T \pi_R(x)}{P_0^T \pi_R(X_t)}-1\right|\leq   \frac{2}{J} \left(\frac{c}{\pi_R(X_t)}e^{-\alpha(T-t)}+\frac{c}{\pi_R(x)}e^{-\alpha T}\right)+1.
\eel{bound_ratio}
Since for each $t$ the ${\mathscr F}_t$-measurable random variable $\frac{1}{\pi_R(X_t)}$ is integrable under $\mathbb{Q}^{\pi_R}_x$ for each $x\in E$, for each $t$  the ${\mathscr F}_t$-measurable random variable $\left|e^{-\lambda_R t}\frac{P_t^T \pi_R(x)}{P_0^T \pi_R(X_t)}-1\right|$ is bounded by an integrable random variable. Furthermore,
by Eq.\eqref{markov_bound1} and \eqref{markov_bound2},
\be
\lim_{T\rightarrow\infty}\left|e^{-\lambda_R t}\frac{P_t^T(\omega) \pi_R(x)}{P_0^T \pi_R(X_t(\omega))}-1\right|=0
\ee
for each $\omega$.
Thus, by the Dominated Convergence Theorem Eq.\eqref{L1_alter} is verified with $B_t^\infty=e^{\lambda_R t}\frac{\pi_R(X_t)}{\pi_R(X_0)}$. $\Box$\\

\bibliographystyleappendix{plainnat}
\bibliographyappendix{mybib7}

\end{document}